\documentclass[prl,amsmath,twocolumn,nofootinbib]{revtex4}

\usepackage{float}
\usepackage{amsmath}
\usepackage{amssymb}
\usepackage{mathtools}
\usepackage{subfigure}
\usepackage{epsfig}
\usepackage{graphicx}
\usepackage{graphics}
\usepackage{epstopdf}

\newcommand{\elabel}[1]{\label{eq:#1}}

\newcommand{\Pref}[1]{Proposition~\ref{prop:#1}}
\newcommand{\pref}[1]{Prop.~\ref{prop:#1}}

\newcommand{\flabel}[1]{\label{fig:#1}}
\newcommand{\fref}[1]{Fig.~\ref{fig:#1}}

\newcommand{\be}{\begin{equation}}
\newcommand{\ee}{\end{equation}}
\newcommand{\bea}{\begin{eqnarray*}}
\newcommand{\eea}{\end{eqnarray*}}

\newcommand{\bi}{\begin{itemize}}
\newcommand{\ei}{\end{itemize}}

\newtheorem{proposition}{Proposition}
\newenvironment{proof}[1][Proof]{\noindent\textbf{#1.} }{\ \rule{0.5em}{0.5em}}

\numberwithin{equation}{section}

\begin{document}
\title{Intergenerational mobility measures in a bivariate normal model}
\author{Yonatan Berman}
\email{yonatanb@post.tau.ac.il}
\affiliation{School of Physics and Astronomy, Tel-Aviv University, Tel-Aviv, Israel}

\date{\today}

\begin{abstract} 
We model the joint log-income distribution of parents and children and derive analytic expressions for canonical relative and absolute intergenerational mobility measures. We find that both types of mobility measures can be expressed as a function of the other.
\end{abstract}

\maketitle

For the past several decades, many scholars have been studying economic intergenerational mobility~\cite{mazumder2005fortunate,aaronson2008intergenerational,lee2009trends,hauser2010intergenerational,corak2013income,chetty2014united}. The motivation for studying mobility stems from its relationship to concepts like equality of opportunity~\cite{roemer2000opportunity,chetty2014land}, the so-called ``American Dream''~\cite{corak2009chasing,chetty2017fading} and income inequality~\cite{corak2013income,berman2016understanding}. Typically measures of income intergenerational mobility are divided into two categories: relative -- quantifying the propensity of individuals to change their position in the income distribution, and absolute -- quantifying their propensity to change their income in money terms. The aim this note is to introduce a simple model for the joint income distribution of parents and children and use it for explicitly deriving canonical measures of relative and absolute mobility measures.

Our starting point is a population of $N$ parent-child pairs. We denote by $Y_p^i$ and $Y_c^i$ the incomes of the parent and the child (at the same age), respectively, for family $i=1\dots N$. We assume the incomes are all positive and move to define the log-incomes $X_p^i=\log Y_p^i$ and $X_c^i=\log Y_c^i$.

The canonical measure of relative mobility is the elasticity of child income with respect to parent income, known as the intergenerational earnings elasticity (IGE)~\cite{mulligan1997parental,lee2009trends,chetty2014land} and defined as the slope ($\beta$) of the linear regression

\be
X_c = \alpha + \beta X_p + \epsilon\,,
\ee
where $\alpha$ is the regression intercept and $\epsilon$ is the error term.

We note that IGE is a measure of immobility rather than of mobility and the larger it is, the stronger the relationship between the parent and child income. Therefore, $1-\beta$  can be used as a measure of mobility.

A standard approach to measure absolute intergenerational mobility, recently used in~\cite{chetty2017fading} for studying the trends in absolute mobility in the United States is to measure the fraction of children earning more than their parents, denoted by $A$:

\be
A = \frac{\sum_{j=1}^{N}{1_{\left\{i:Y_c^i>Y_p^i\right\}}\left(Y_c^j\right)}}{N}\,,
\ee
where $1_S\left(x\right)$ is the indicator function for a set $S$ and argument $x$ and $\left\{i:Y_c^i>Y_p^i\right\}$ is the set of children earning more than their parents.

Since the logarithmic function preserves order we also get,

\be
A = \frac{\sum_{j=1}^{N}{1_{\left\{i:X_c^i>X_p^i\right\}}\left(X_c^j\right)}}{N}\,.
\ee

One hypothetical sample of such distribution is presented in~\fref{lines}. It also depicts graphically how $A$ and $\beta$ are defined. The blue line is $y=x$, hence the rate of  absolute mobility is defined as the fraction of parent-child pairs which are above it. The red line is the linear regression $y=\alpha +\beta x$, for which $\beta$ is the IGE.

\begin{figure}[!htb]
\centering
\includegraphics[width=0.5\textwidth]{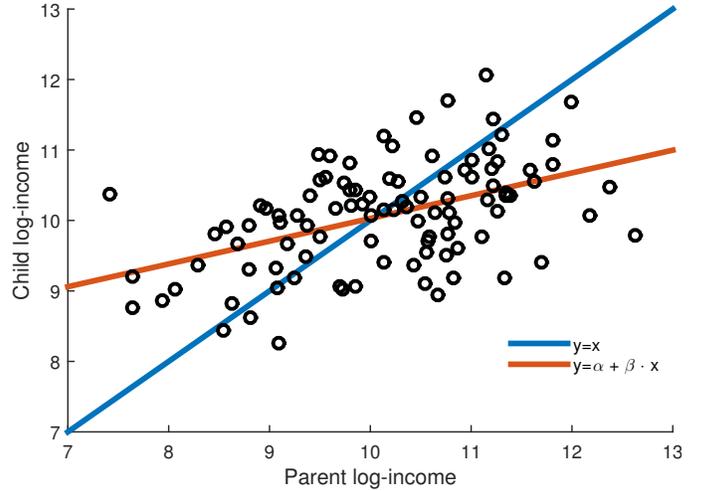}
\caption{An illustration of the absolute and relative mobility measures. The black circles are a randomly chosen sample of 100 parent-child log-income pairs. The sample was created assuming a bivariate normal distribution and the parameters used were $\mu_p=10.1$, $\sigma_p=0.78$ (for the parents marginal distribution) and $\mu_c=10.25$, $\sigma_c=1.15$ (for the children marginal distribution) with correlation of $\rho=0.57$. The resulting $\alpha$ and $\beta$ were $1.8$ and $0.84$, respectively.}
\flabel{lines}
\end{figure}

Since income distributions are known to be well approximated by the log-normal distribution~\cite{pinkovskiy2009parametric}, a simple plausible model for the joint distribution of parent and child log-incomes is the bivariate normal distribution. Under this assumption, the marginal income distributions of both parents and children are log-normal and the correlation between their log-incomes is defined by a single parameter $\rho$. The marginal log-income distribution of the parents (children) follows $\mathcal{N}\left(\mu_p,\sigma_p^2\right)$ ($\mathcal{N}\left(\mu_c,\sigma_c^2\right)$), hence the joint distribution is fully characterized by 5 parameters: $\mu_p$, $\sigma_p$, $\mu_c$, $\sigma_c$ and $\rho$.

Assuming the bivariate normal approximation for the joint distribution enables theoretically studying its properties. In particular Both measures of mobility, $A$ and $1-\beta$, can be derived directly from the model and, notably, can both be expressed analytically as functions of the other. We first derive the IGE in terms of the distribution parameters:

\begin{proposition}
\label{prop:prop1}

For a bivariate normal distribution with parameters $\mu_p$, $\sigma_p$ (for the parents marginal distribution) and $\mu_c$, $\sigma_c$ (for the children marginal distribution) assuming correlation $\rho$, the IGE is

\be
1-\beta = 1-\frac{\sigma_c}{\sigma_p}\rho \,.
\elabel{beta_rho}
\ee
\end{proposition}

\begin{proof}
First, by definition, the correlation $\rho$, between $X_p$ and $X_c$ equals to their covariance, divided by $\sigma_p\sigma_c$

\be
\rho = \frac{\text{Cov}\left[X_p,X_c\right]}{\sigma_p\sigma_c}\,.
\ee

$\beta$ can be directly calculated as follows, by the linear regression slope definition:

\be
\beta = \frac{\sum_{i=1}^{N} {\left(X_p^i - \bar{X}_p\right)\left(X_c^i - \bar{X}_c\right)}}{\sum_{i=1}^{N} {\left(X_p^i - \bar{X}_p\right)}}\,,
\ee
where $\bar{X}_p$ and $\bar{X}_c$ are the average parents and children log-incomes, respectively.

It follows that 
\be
\beta = \frac{\text{Cov}\left[X_p,X_c\right]}{\sigma_p^2}\,.
\ee

We immediately obtain

\be
\beta = \frac{\sigma_c}{\sigma_p}\rho
\ee

and therefore

\be
1-\beta = 1-\frac{\sigma_c}{\sigma_p}\rho
\ee
\end{proof}

Following \pref{prop1} it is also possible to derive the rate of absolute mobility as a function of the distribution parameters and the IGE:

\begin{proposition}
\label{prop:prop2}

For a bivariate normal distribution with parameters $\mu_p$, $\sigma_p$ (for the parents marginal distribution), $\mu_c$, $\sigma_c$ (for the children marginal distribution) and $\rho=\sigma_p\beta/\sigma_c$ (where $\beta$ is the IGE), the rate of absolute mobility is

\be
A = \Phi\left(\frac{\mu_c - \mu_p}{\sqrt{\sigma_p^2\left(1 - 2\beta\right) + \sigma_c^2}}\right) \,,
\elabel{abs2}
\ee
where $\Phi$ is the cumulative distribution function of the standard normal distribution.
\end{proposition}

\begin{proof}
We start by defining a new random variable $Z = X_c-X_p$. It follows that calculating $A$ is equivalent to calculating the probability $P\left(Z>0\right)$.

Subtracting two dependent normal distributions yields that $Z \sim \mathcal{N}\left(\mu_c - \mu_p,\sigma_p^2 + \sigma_c^2 - 2\text{Cov}\left[X_p,X_c\right]\right)$, so according to \pref{prop1}

\be
Z \sim \mathcal{N}\left(\mu_c - \mu_p,\sigma_p^2\left(1-2\beta\right) + \sigma_c^2\right)\,.
\ee

If follows that

\be
\frac{Z - \left(\mu_c - \mu_p\right)}{\sqrt{\sigma_p^2\left(1-2\beta\right) + \sigma_c^2}} \sim \mathcal{N}\left(0,1\right)\,,
\ee

so we can now write

\be
\begin{split}
&P\left(Z>0\right) = \\ & P\left(\frac{Z - (\mu_c - \mu_p)}{\sqrt{\sigma_p^2\left(1-2\beta\right) + \sigma_c^2}} > -\frac{\mu_c - \mu_p}{\sqrt{\sigma_p^2\left(1-2\beta\right) + \sigma_c^2}} \right) = \\ &\Phi\left(\frac{\mu_c - \mu_p}{\sqrt{\sigma_p^2\left(1 - 2\beta\right) + \sigma_c^2}}\right) \,,
\end{split}
\ee
where $\Phi$ is the cumulative distribution function of the standard normal distribution.

\end{proof}

\Pref{prop2} shows that the rate of absolute mobility can be explicitly described as a function of the relative mobility.

\bibliographystyle{plain}
\bibliography{mobmob}

\end{document}